\documentclass[12pt]{article}
\usepackage{graphicx,amssymb,amsmath}

\usepackage{microtype}%if unwanted, comment out or use option "draft"

\usepackage{subcaption}
\usepackage{caption,setspace}

\usepackage{latexsym}
\usepackage{wrapfig}
\usepackage{algorithm}
\usepackage{amssymb, amsmath}
\usepackage[noend]{algcompatible}
\usepackage{amsmath}
\usepackage{mathrsfs}
\usepackage{verbatim}
\usepackage{breqn}
\usepackage{color}
\usepackage{yhmath}
\usepackage{afterpage}
\usepackage{epsfig}
\usepackage{etoolbox}
\usepackage{color,soul}
\usepackage[svgnames]{xcolor}

\title{Time-Dependent Shortest Path Queries Among Growing Discs}

\author{
  Anil Maheshwari
    \and
    Arash Nouri
    \and
    J\"org-R\"udiger Sack 
    }

\date{School of Computer Science, Carleton University\\
\{anil,arash,sack\}@scs.carleton.ca}

\begin{document}

\renewcommand{\algorithmiccomment}[1]{\color{mygray} \hfill $\triangleright$ \textit{#1} \color{black}}

\definecolor{mygray}{gray}{0.3}

\newtheorem{innercustomthm}{\textbf{Lemma}}
\newenvironment{customthm}[1]
  {\renewcommand\theinnercustomthm{#1}\innercustomthm}
  {\endinnercustomthm}

\newtheorem{corollary}{Corollary}
\newtheorem{rep@theorem}{\rep@title}
\newcommand{\newreptheorem}[2]{%
\newenvironment{rep#1}[1]{%
 \def\rep@title{#2 \ref{##1}}%
 \begin{rep@theorem}}%
 {\end{rep@theorem}}}
\makeatother

\newtheorem{theorem}{Theorem}
\newtheorem{lemma}{Lemma}
\newenvironment{proof}           
{
  \noindent{\bf Proof.} 
}
{
  \hfill$\Box$\newline
}

\newtheorem{observation}{\textbf{Observation}}
\newtheorem{property}{\textbf{Property}}
\newtheorem{algX}{Algorithm}
\newenvironment{Algorithm}       {\begin{algX}\begin{em}\begin{sf}}%
                                 {\end{sf}\par\noindent
                                 --- End of Procedure ---
                                 \end{em}\end{algX}}
\newcommand{\step}[2]            {\begin{list}{}
                                  {  \setlength{\topsep}{0cm}
                                     \setlength{\partopsep}{0cm}
                                     \setlength{\leftmargin}{0.55cm}
                                     \setlength{\labelwidth}{0.55cm}
                                     \setlength{\labelsep}{0.1cm}    }
                                  \item[#1]#2    \end{list}}

\newcommand{\sstep}[2]            {\begin{list}{}
                                  {  \setlength{\topsep}{0cm}
                                     \setlength{\partopsep}{0cm}
                                     \setlength{\leftmargin}{1cm}
                                     \setlength{\labelwidth}{0.55cm}
                                     \setlength{\labelsep}{0.1cm}    }
                                  \item[#1]#2    \end{list}}

\newcommand{\mystep}[2]            {\begin{list}{}
                                  {  \setlength{\topsep}{0cm}
                                     \setlength{\partopsep}{0cm}
                                     \setlength{\leftmargin}{1.4cm}
                                     \setlength{\labelwidth}{0.55cm}
                                     \setlength{\labelsep}{0.1cm}    }
                                  \item[#1]#2    \end{list}}

\renewcommand{\theequation}{A-\arabic{equation}}
  \setcounter{equation}{0} 

\newcommand{\ee}{\mathcal{E}}
\newcommand{\vc}[1]{V(#1)}
\renewcommand{\v}[1]{\mathcal{V}_{#1}}
\newcommand{\dc}[1]{DC({#1})}
\renewcommand{\c}[3]{C_{#2}({#1})}
\renewcommand{\l}[3]{
  \ifstrequal{#1}{}
    {\ell^{#3}_{#2}}
    {\ell^{#3}_{#2}({#1})}
}

\newcommand{\location}{\mathcal{L}}

\newcommand{\departurecurve}{\mathcal{D}}

\newcommand{\isvalid}{\mathcal{V}}

\newcommand{\steinerpoints}{\mathcal{S}}

\newcommand{\alignset}{\mathcal{L}}

\newcommand{\weight}{\mathcal{W}}
\newcommand{\weights}[2]{\mathcal{W}({#1}, #2)}

\newcommand{\invalidinterval}{\mathcal{B}}

\newcommand{\robotpath}{\mathcal{R}}

\newcommand{\minimumarrival}{\mathcal{A}}

\newcommand{\T}{\lambda}

\newcommand{\au}[1]{\mathcal{A}(#1)}
\renewcommand{\ae}[2]{A_{#2}(#1)}

\newcommand{\aus}[2]{\overline{A}(t)}
\newcommand{\aes}[3]{\overline{A}(#3)}

\renewcommand{\L}[2]{L_{(#1, #2]}}
\newcommand{\R}[2]{R_{(#1, #2]}}

\newcommand{\aux}[2]{A_{(#1, #2]}(t)}
\newcommand{\aex}[3]{A_{(#1, #2]}(#3)}

\newcommand{\tu}[1]{T_{s,d}(#1)}
\newcommand{\te}[2]{T_{#2}(#1)}

\newcommand{\fus}[2]{\overline{Aprx}(t)}
\newcommand{\fes}[3]{\overline{Aprx}({#3})}

\newcommand{\fex}[3]{Aprx_{(#1, #2]}({#3})}
\newcommand{\fux}[2]{Aprx_{(#1, #2]}(t)}

\newcommand{\fu}[1]{Aprx_{s,d}(#1)}
\newcommand{\fe}[2]{Aprx_{#2}(#1)}

\newcommand{\vm}{\mathcal{V}_{c}}
\newcommand{\vr}{\mathcal{V}_{r}}

\newcommand{\dd}{\beta}

\renewcommand{\ae}[2]{\mathcal{A}_{#2}(#1)}

\newcommand{\aae}{\overline{\mathcal{A}}}

\renewcommand{\algorithmiccomment}[1]{\color{mygray} \hfill $\triangleright$ \textit{#1} \color{black}}

\newcommand{\aaa}{\mathcal{A}}

\maketitle

\begin{abstract}
The determination of time-dependent collision-free shortest paths has received a fair amount of attention.
Here, we study the problem of computing a time-dependent shortest path among growing discs which has been previously studied for the instance where the departure times are fixed.
We address a more general setting: For two given points $s$ and $d$, we wish to determine the function $\aaa(t)$ which is the minimum arrival time at $d$ for any departure time $t$ at $s$.
We present a $(1+\epsilon)$-approximation algorithm for computing $\aaa(t)$. 

As part of preprocessing, we execute $O({1 \over \epsilon} \log({\vr \over \vm}))$ shortest path computations for fixed departure times, where $\vr$ is the maximum speed of the robot and $\vm$ is the minimum growth rate of the discs. 
For any query departure time $t \geq 0$ from $s$, we can approximate the minimum arrival time at the destination in 
$O(\log ({1 \over \epsilon}) + \log\log({\vr \over \vm}))$ 
 time, within a factor of $1+\epsilon$ of optimal. Since we treat the shortest path computations as black-box functions, for different settings of growing discs, we can plug-in different shortest path algorithms. Thus, the exact time complexity of our algorithm is determined by the running time of the shortest path computations.
 \end{abstract}

\section{Introduction }

 An algorithmic challenge in robotics arises when a point object (modeling a robot, person or vehicle) is operating among moving entities or obstacles, e.g., settings in which a point object needs to avoid encountering individuals who are moving from known locations, with known speeds, but in unknown directions. This uncertainty can be modeled by discs, whose radii grow over time. Therefore, the task of computing a shortest path avoiding these individuals reduces to computing shortest path among growing discs. This particular motivation arose, e.g., in video games \cite{shortest_path_in_growing_disc_referred}.

Given are a set of growing discs $\mathscr{C}=\{C_1, ..., C_n\}$ (the obstacles), a point robot $R$ with maximum speed $\vr$, a source point $s$ and a destination point $d$, located on the plane. The radii of the discs are growing with the same constant speed $V \in (0, \vr)$.
The \textit{shortest path among growing discs} (SPGD) problem is to find a shortest path from $s$ to $d$, such that the robot leaves the source immediately (i.e., at time $t=0$) and does not intersect the interior of the discs, to reach $d$ as quickly as possible.
Yi \cite{shortest_path_in_growing_disc_yi} showed that this problem can be solved in $O(n^2 \log n)$ time.

In this paper, we study the time-dependent version of the SPGD problem, where the departure time is a variable. 
The objective is to find the minimum arrival time function $\aaa(t)$, defined as the earliest time when the robot can arrive at destination $d$ such that: (1) $R$ leaves $s$ at time $t$, (2) $R$ does not intersect the interior of the discs after the departure.
We refer to this problem as the \textit{time-dependent shortest path among growing discs} (TDSP) problem.

\textbf{Related results.} 
The SPGD problem has been studied in different settings. Overmars et al. studied this problem in the setting where the discs are growing with equal constant speed. They presented an $O(n^3 \log n)$ time algorithm, where $n$ is the number of discs. This result was improved by Yi \cite{shortest_path_in_growing_disc_yi} who showed that the shortest path among the same-speed growing discs can be found in  $O(n^2 \log n)$ time. Yi also presented an $O(n^3 \log n)$ time algorithm for the case where the discs are growing with different speeds. Nouri and Sack \cite{nouri} studied a general version of this problem where the speeds are polynomial functions of time.

Time-dependent shortest path problems have been studied in network settings (see \cite{Dehne-TDSP, masoud-TDSP, suri-TDSP}). 
A \textit{network} is a graph $G=(V, E)$ with edge set $E$ and node set $V$. Each edge $e \in E$ is
assigned with a real valued weight. 
% We assume $G$ is a FIFO network, i.e., the commodities travel along the edges in a First-In-First-Out manner. 
Given a source node $s \in V$ and a destination node $d \in V$, a shortest path from $s$ to $d$ is a path in $G$, where the sum of the weights
of its constituent edges is minimized. However,
in many applications, the weight of the edges are dynamically changing
over time. In such situations, the total weight of a path depends on the departure time at its source. The problem of computing shortest paths from $s$ to $d$ for all possible departure times at
$s$ is known as the time-dependent shortest path problem. The general shortest path problem on time-dependent
networks has been proven to be NP-Hard \cite{orda}. However, there are several approximation algorithms (see \cite{Dehne_FIFO, masoud-TDSP, suri-TDSP}), which are of interest in real-world applications \cite{chen}. 

 \textbf{Contribution.}  We say an approximation function $\aae(t, \epsilon)$ is a $(1+\epsilon)$-approximation for function $\aaa(t)$ if $\aaa(t) \leq \aae(t, \epsilon) \leq (1+\epsilon)\aaa(t)$ for all positive values of $t$. 
Here, our contribution is to compute a $(1+\epsilon)$-approximation for the minimum arrival time function $\aaa(t)$. 
The preprocessing step of our algorithm executes $O({1 \over \epsilon} \log({\vr \over \vm}))$ time-minimal path computations 
% (defined in Section \ref{PTDSP}) 
for fixed departure times, where $\vr$ is the maximum speed of the robot and $\vm$ is the minimum growth rate of the discs. 
Then, for a given query departure time $t \geq 0$ from $s$, we can report the minimum arrival time at the destination in $O(\log ({1 \over \epsilon}) + \log\log({\vr \over \vm}))$  time, within a factor of $1+\epsilon$ of optimal. We first start with a simple version of the problem, where all discs are growing with speed $\vm$. In this setting, each time-minimal path computation runs in $O(n^2 \log n)$ time \cite{shortest_path_in_growing_disc_yi}. Our algorithm runs the shortest path computations as a black-box and its time complexity  is determined by the number of such calls. This enables us to extend our algorithm for different settings of growing discs \cite{shortest_path_in_growing_disc_yi,previous}, by plugging-in appropriate shortest path algorithms.

In Section \ref{prop}, we establish several properties of the arrival time function. These properties allow us to work with arrival time functions instead of a more indirect approach of using the travel time, which has been utilized in previous work. 
% These properties are also of independent interest.
 % and might be useful for designing other algorithms. 
% For each collision-free path $\pi$ between the source and the destination, we define arrival time function $\aaa_{\pi}(t)$ which represents the arrival time of path $\pi$ when the departure time is $t$. Then, 
 We show that $\aaa(t)$ is the lower envelope of a set of curves called the arrival time functions. The algorithm's output size is denoted by $F_{\aaa}$ and counts the number of pieces 
(i.e. the sub-arcs) 
on the lower envelope needed to represent the
 function $\aaa(t)$. In Section \ref{ouput}, we establish a lower bound for $F_{\aaa}$. This lower bound, along with the complexity of computing the lower envelope, provides the motivation to study the approximation algorithm in the first place.

In Section \ref{appA}, we define the \textit{reverse shortest path problem}, in which we are to find a path from the destination to the source.
Existing algorithms for the time dependent shortest paths utilize the reverse shortest path computations. 
% In previous approaches, 
These computations 
were done by a reversal of Dijkstra's algorithm \cite{Reversibility,masoud-TDSP}. Here, we need to generalize the existing shortest path computations for growing discs \cite{shortest_path_in_growing_disc_yi} to shortest path computations for shrinking discs.

\section{Preliminaries}\label{prel}

\subsection{Time-minimal paths among growing discs} \label{PTDSP}
 
A \textit{robot-path} is a path in the plane that connects the source point $s$ to the destination point $d$. 
% This path can be specified using a function $\T:[T_s, T_d] \rightarrow {\rm I\!R}^2$, where $\T(t)$ returns the location of the robot at any time $t \in [T_s, T_d]$, provided that $\T(T_s)=s$ and $\T(T_d)=d$. 
% We call $T_s$ the \textit{departure time} and $T_d$ the \textit{arrival time}.
The time at which the robot departs $s$ is called the \textit{departure time} and the time it arrives at $d$ is the \textit{arrival time}.
% The robot-path $\T$ \textit{intersects} disc $C_i$, if there exists a time $t \in [T_s, T_d]$, such that $\T(t)$ returns a point located in the interior of disc $C_i$. 
  We call a path $\T$ \textit{valid} (or \textit{collision-free}) if it does not intersect any of the discs in $\mathscr{C}$; otherwise, $\T$ is \textit{invalid}. 
For a fixed departure time, a \textbf{time-minimal path} is a valid robot-path, where the arrival time is minimized over all valid paths.

It is proven that on any time-minimal path, the point robot always moves with maximal velocity of $\vr$ \cite{shortest_path_in_growing_disc_referred}.
It is shown in \cite{nouri, growing_discs_overmars} that any time-minimal path from $s$ to $d$ is solely composed of two types of alternating sub-paths: (1) \textbf{tangent paths:} straight line paths that are
tangent to pairs of discs, and (2) \textbf{spiral paths:} logarithmic spiral paths that each lies on a
boundary of the growing disc. We describe these two paths in the following.

\begin{figure}[t]
\captionsetup[subfigure]{justification=centering}
\centering
    \begin{subfigure}{0.4\linewidth}
    \centering
    \includegraphics[width=100pt]{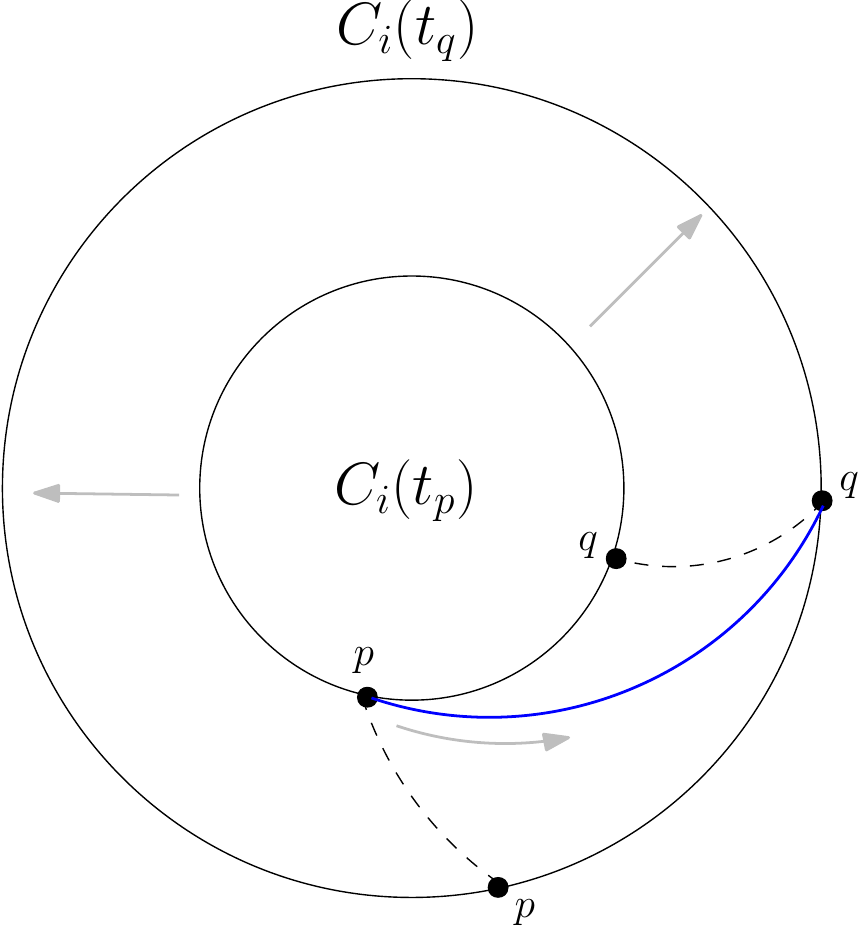}
  \caption{}
  \end{subfigure}
    \begin{subfigure}{0.4\linewidth}
    \centering 
    \includegraphics[width=70pt]{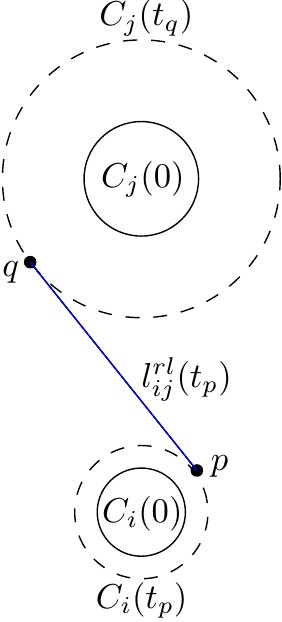}
    \caption{}
  \end{subfigure}
\caption{Two robot-paths are illustrated: (a) spiral path, (b) a right-left tangent path. Note that $t_p$ is the departure time and $t_q$ is the arrival time where $t_q > t_p$. }\label{fig-spiral}
\end{figure}

For any pair of discs $C_i$ and $C_j$, we define four \textit{tangent paths} corresponding to their tangent lines: right-right,
right-left, left-right and left-left tangents, denoted by $\ell_{ij}^{rr}$, $\ell_{ij}^{rl}$, $\ell_{ij}^{rl}$ and $\ell_{ij}^{ll}$ (see Figure \ref{fig-spiral}). 
Each tangent path $\overrightarrow{\ell_{ij}^{rr}}(\tau) = \overline{pq}$ represents a straight line robot-path from a point $p$ on the boundary of $C_i$ to a point $q$ on the boundary of $C_j$. The two points $p$ and $q$ are called \textit{Steiner points}.
A \textit{spiral path} $\overrightarrow{\sigma}=\wideparen{pq}$, represents the trace of the robot's move, over
time, from $p$ to $q$, where $p,q \in \partial C_i$ are two ``consecutive'' Steiner points (see Figure \ref{fig-spiral} (a)). 
Note that the length of these paths are changing over time. If the robot leaves $p$ at time $\tau$, then it arrives at $q$ at time $\hat{\tau}$, where $\tau < \hat{\tau}$. Observe that there exist $O(n^2)$ (moving) tangent/spiral paths and $O(n^2)$ (moving) Steiner points.

To simplify our exposition, we let the two points $s$ and $d$ be two discs with zero radii and zero velocities and add them to $\mathscr{C}$.
Let $\ee$ be the set of all spiral and tangent paths. Let $\steinerpoints$ be the set of all Steiner points. 
We construct a directed \textit{adjacency graph} $G = (V_s, E_s)$ as follows. With each Steiner point $v \in \steinerpoints$ we associate a unique vertex, $\dot{v}$, in $V_s$.
Then, with each path $\overrightarrow{uv} \in \ee$ we associate a unique edge $\overrightarrow{\dot{u}\dot{v}}$ in $E_s$. 
% For ease of notation, we denote $\dot{u}\dot{v}$ by $\overrightarrow{e}$. 
The weight of each edge in $G$ is the length of its corresponding tangent or spiral path, which is a function of time.
% function for which the inputs are $e$ and $\tau$, and the output is the length of the path $\robotpath(e, \tau)$ if it is valid, and $\infty$ otherwise.
Since each edge in $E_s$ is associated with a path in $\ee$, therefore, each path in the graph $G$ is associated with a sequence of paths in $\ee$.  

Yi \cite{shortest_path_in_growing_disc_yi} showed that by running Dijkstra's algorithm on the adjacency graph, the SPGD problem can be solved in $O(n^2 \log n)$ time. The main steps of this algorithm are as follows:

\begin{itemize}
\item[($i$)] Identify the Steiner points, tangent paths and spiral paths (i.e., the sets $\steinerpoints$ and $\ee$).
\item[($ii$)] Construct the adjacency graph $G = (V_s, E_s)$.
\item[($iii$)] Run Dijkstra's algorithm to find a time-minimal path between $s,d \in V_s$.
\end{itemize}

In our approximation algorithm, we use the above algorithm as a black-box. The input to this algorithm is a set of growing discs and a departure time $\tau$. The output is a shortest path between $s$ and $d$ in the adjacency graph, denoted by $\pi(s, d, \tau)$.

\subsection{Minimum time-dependent arrival time}

\begin{figure}
\captionsetup[subfigure]{justification=centering}
\centering
  \begin{subfigure}{0.3\linewidth}
  \centering
  \includegraphics[width=100pt]{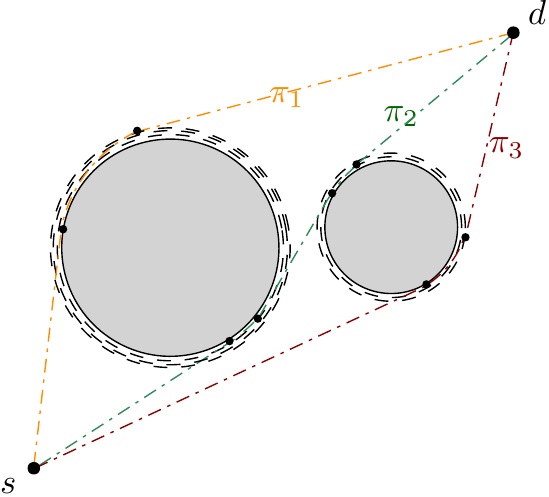}
    \caption{}
  \end{subfigure}
  \begin{subfigure}{0.3\linewidth}\centering
  \includegraphics[width=110pt]{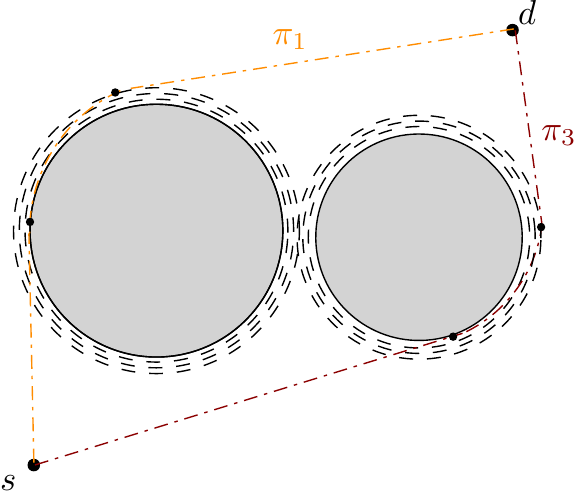}
    \caption{}
  \end{subfigure}
  \begin{subfigure}{0.3\linewidth}\centering
  \includegraphics[width=120pt]{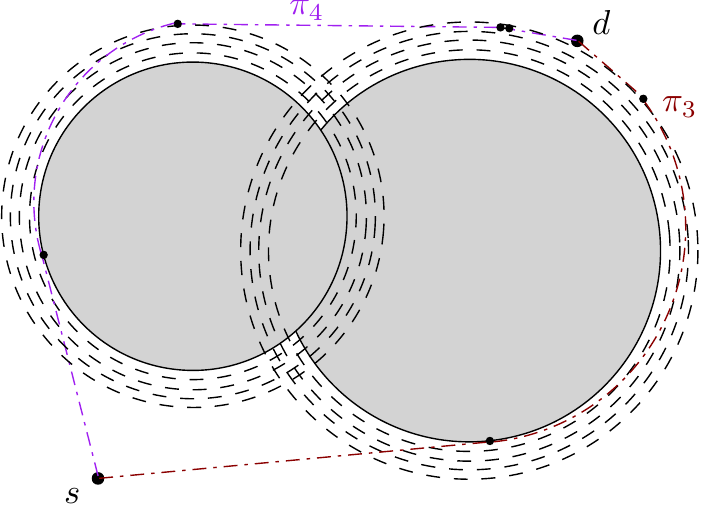}
    \caption{}
  \end{subfigure}
\caption{This figure illustrates valid paths between $s$ and $d$ for three different departure times. 
(a) At this time, three paths $\pi_1$, $\pi_2$ and $\pi_3$ are valid. (b) Path $\pi_2$ is obstructed and becomes invalid. Consequently, path $\pi_3$ is the time-minimal path. (c) Path $\pi_1$ is obstructed and $\pi_4$ is the time-minimal path.}\label{arrival-maximum-paths}

\end{figure} 

\begin{figure}[t]
\centering
\includegraphics[width=200pt]{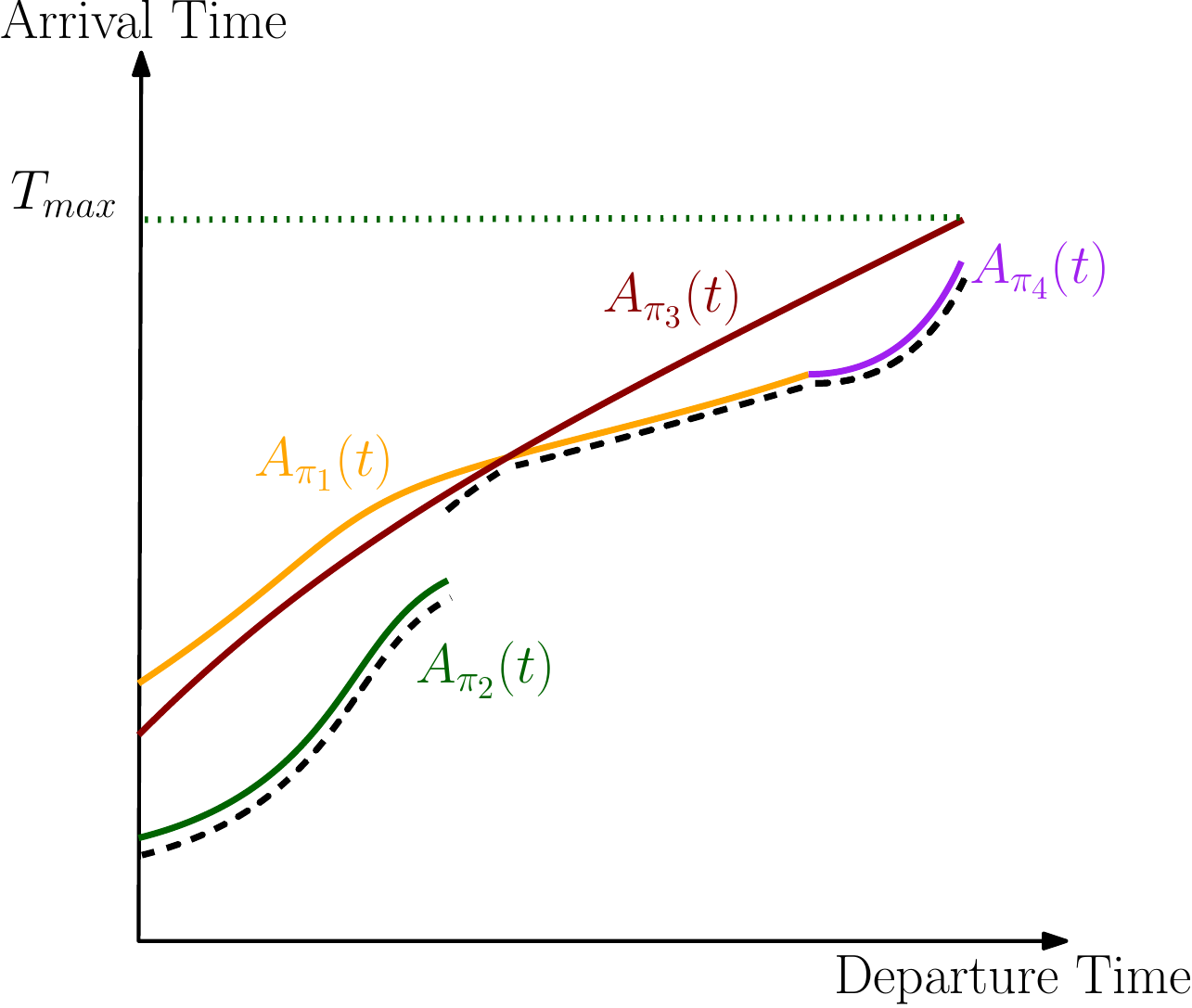}
\caption{This figure represents the arrival curves corresponding to the time-minimal paths in Figure \ref{arrival-maximum-paths}. The dotted curve (the lower envelope) represents the minimum arrival time function $\aaa(t)$. 
% [*LE-Curve*]
}\label{minimum-arrival-time}
\end{figure}

 Let $\mathscr{P}$ be the set of all paths between $s$ and $d$, in the adjacency graph $G$. Let $\pi \in \mathscr{P}$ and $\tau$ be a departure time at $s$. Recall that there exists a unique geometric path corresponding to the pair $(\pi, \tau)$.
  For instance, Figure \ref{arrival-maximum-paths} shows the geometric paths corresponding to a set of paths $\{\pi_1, \pi_2, \pi_3, \pi_4\}$ at three different departure times.

For a given path $\pi \in \mathscr{P}$, let the \textit{arrival time function} $\ae{t}{\pi}$ represents the arrival time of the robot-path corresponding to $\pi$ when the departure time is $t$. 
Let $\mathscr{A}=\{\ae{t}{\pi}|\pi\in \mathscr{P}\}$ be the arrival time function corresponding to the paths in $\mathscr{P}$. The lower envelope of the set $\mathscr{A}$, denoted by $\aaa(t)$, is defined as the point-wise minimum of all the arrival functions in $\mathscr{A}$. The lower envelope is formally defined as:
\begin{align*}\aaa(t) = \min\limits_{\pi \in \mathscr{P}} \overline{A_{\pi}}(t)\end{align*}
where $\overline{A_{\pi}}(t) = \ae{t}{\pi}$ if $\pi$ is a valid path and $\overline{A_{\pi}}(t) = \infty$, otherwise. 
We call $\aaa(t)$ the \textit{minimum arrival time function}, which is found via finding the lower envelope of the arrival functions in $\mathscr{A}$. We also refer to the minimum arrival time functions as arrival curves.

Figure \ref{minimum-arrival-time} depicts the arrival functions corresponding to the paths in Figure \ref{arrival-maximum-paths}. 
 $\aaa(t)$ is found as the lower envelope of the arrival functions $\{\ae{t}{\pi_1}, \ae{t}{\pi_2}, \ae{t}{\pi_3}, \ae{t}{\pi_4}\}$. 
 Observe that after a certain time, the destination point will be contained in some disc. We denote this time by $T_{max}$. Since there exists no valid path from $s$ to $d$ after $T_{max}$, for any $t > T_{max}$ we have $\aaa(t) = \infty$.
  Observe that the lower envelope is composed of sub-arcs of arrival curves in $\mathscr{A}$. A \textit{sub-arc} is a maximal connected piece of an arrival curve on the lower envelope.
 % We refer to these pieces as \textbf{sub-arcs}. 
 % By increasing the departure time from $0$ to $T_{max}$, the minimum arrival time may move from one sub-arc to another (see Figure \ref{minimum-arrival-time}). 

\section{Properties of $\aaa(t)$} \label{prop}

% These properties will be
% used in Section \ref{appA}. 

In this section, we define some properties of the minimum arrival time function.
First, we show that $\aaa(t)$ is an increasing function. If the robot is allowed to move with any speed lower than $\vr$, then the below lemma is straightforward (the later it departs from the source the later it arrives at the destination). However, we assumed that the robot moves with maximum speed at all time and only uses tangent and/or spiral paths. Thus, the following proof is necessary.

\begin{figure}[t]
\centering
\includegraphics[width=220pt]{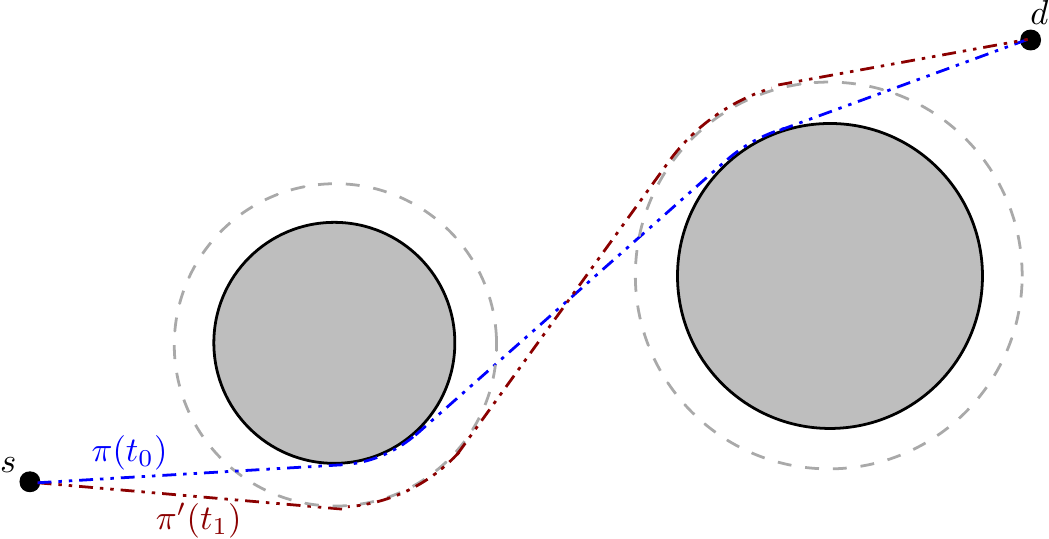}
\caption{ This figure illustrates two time-minimal paths for two departure times $t_0$ and $t_1$, where $t_1 > t_0$. For the sake of simplicity, we assumed that the speed of the robot is considerably higher than the growth rates of the discs ($V \ll \vr$).}\label{property-i}
\end{figure}

\begin{lemma}\label{T-properties-1}
$\aaa(t)$ is an increasing function.
\end{lemma}

\begin{proof} Let $\pi({t_0})$ and $\pi'({t_1})$ be any two time-minimal robot-paths between $s$ and $d$, corresponding to two departure times $t_0$ and $t_1$, where $t_0 < t_1$ (see Figure \ref{property-i}). We prove that $\ae{t_0}{\pi} < \ae{t_1}{\pi'}$ and consequently $\aaa(t_0) < \aaa(t_1)$. 

Let $\te{t_1}{\pi'} = \ae{t_1}{\pi'} - t_1$ be the travel time for the path $\pi'(t_1)$. Similarly, let $\te{t_0}{\pi} = \ae{t_0}{\pi} - t_0$. By contradiction, we assume $\ae{t_1}{\pi'} \leq \ae{t_0}{\pi}$. Then:
\begin{align*}
&\te{t_1}{\pi'} + t_1 \leq \te{t_0}{\pi} + t_0 \\
\stackrel{t_0 < t_1}{\Rightarrow}&\te{t_1}{\pi'}  < \te{t_0}{\pi} \\
\end{align*}
 Therefore, $\pi({t_0})$ is a longer robot-path than $\pi'({t_1})$. Since the discs are continuously growing, the free space is shrinking simultaneously. Thus, if $\pi'({t_1})$ is a valid robot-path, it is also valid for any departure time before $t_1$. As illustrated in Figure \ref{property-i}, $\pi'({t_1})$ is a valid path when the robot leaves $s$ at time $t_0$.

Since $|\pi'({t_1})| < |\pi({t_0})|$ and $\pi'({t_1})$ is a valid robot-path when the robot departs $s$ at time $t_0$, this contradicts the fact that $\ae{t_1}{\pi'} \leq \ae{t_0}{\pi}$.
% $\pi({t_0})$ is a time-minimal robot-path at time $t_0$. 
\end{proof}

Define $|\overline{sd}|$ as the Euclidean distance between the source and the destination. Recall that $\vr$ is the maximum speed of the robot, and $\vm$ is the minimum growth rate among the discs in $\mathscr{C}$. 

Let $\overrightarrow{sd}$ be a tangent path from $s$ to $d$.
For a given departure time $t$, if $\overrightarrow{sd}(t)$ is not obstructed by any disc, then the calculation of $\aaa(t)$ is straightforward. Thus, we are interested in computing $\aaa(t)$ when $\overrightarrow{sd}(t)$ is invalid, i.e., 
% Hence, without loss of generality, assume that 
$\overrightarrow{sd}(t)$ is intersected by a disc for any departure time $t$. With the above assumption, we have the following lemma which states the upper and the lower bound on $\aaa(t)$.
 
% Using these notations, we define and prove two properties in the following lemmas, which are subsequently relevant to the arrival time functions.

\begin{lemma} \label{T-properties-2}
Let $\tau$ be a departure time, where $\au{\tau}$ is defined (i.e., $\aaa(\tau) < \infty$). Then, 

 \begin{enumerate}
  \item[(i)] $\au{\tau} \geq {|\overline{sd}| \over \vr}$

  \item[(ii)] $\au{\tau} \leq {|\overline{sd}| \over \vm}$

  % \item[(iv)] $T_{\pi}(t)=\infty$ at any time $t \in (T_{dom} - {L \over \vr}, \infty)$, where $T_{dom}$ is earliest time where $d \in C_i$.
 \end{enumerate}

% \hlr{should I use $\tau$ or $t$? Since $\tau$ is given, I thought I should make a distinction between variable $t$ and $\tau$}

\end{lemma}

% \begin{figure}[t]
% \centering
% \includegraphics[width=200pt]{fig-arrival-maximum-eps-converted-to.pdf}
% \caption{Path $\overline{sd}$ intersected by disc $C \in \mathscr{C}$.}\label{arrival-maximum}
% \end{figure}

\begin{proof}
(i)  
Let $\T_1$ be a valid time-minimal robot-path from $s$ to $d$ with departure time $0$. 
It is observed that the length of $\T_1$ is greater or equal to $|\overline{sd}|$. Thus, ${|\overline{sd}| \over \vr} \leq \au{0}$.
By Lemma \ref{T-properties-1}, we have $\aaa(0) \leq \aaa(\tau)$.
Therefore, $  {|\overline{sd}| \over \vr} \leq \au{\tau}$.

 (ii) Let  $\T_2$ be the straight line (invalid) robot-path from $s$ to $d$, which is obstructed by the disc $C \in \mathscr{C}$. Let the robot depart $s$ at time $\tau$ and move along the path $\T_2$ with speed $\vr$, until it arrives at the boundary of $C$ at point $q$. Note that the robot arrives at $q$ at time $\tau + {|\overline{sq}| \over \vr}$.
Let $x$ be the shortest Euclidean distance between $d$ and the boundary of $C$ at time $\tau + {|\overline{sq}| \over \vr}$. If $C$ encloses $d$ at time $T$, then the robot must arrive at the destination at or before $T$.
Thus, we obtain $T \leq \tau + {|\overline{sq}| \over \vr} + {x \over V}$.
Because $x \leq |\overline{qd}|$ and $V < \vm$, we have:
  \begin{align*}
  T \leq \tau + {|\overline{sq}| \over \vr} + {|\overline{qd}| \over \vm}
  \end{align*}
  Since the above inequality is true for all departure times (including $\tau = 0$), we must have: 
 \begin{align*}
 &T \leq {{|\overline{sq}| \over \vm} + {|\overline{qd}| \over \vm}} = {|\overline{sd}| \over \vm}\\
 \Rightarrow& \aaa(\tau) \leq {|\overline{sd}| \over \vm}
\end{align*}
\end{proof}

\subsection{The output size} \label{ouput}
The \textit{output size} of the time-dependent shortest path, denoted by $F_{\aaa}$, is defined as the number of sub-arcs in the lower-envelope needed to represent the minimum arrival time function $\aaa(t)$. 
% The description complexity of $\aaa(t)$, denoted by $F_{\aaa}$, is the sum of the description complexities of all sub-arcs in the lower-envelope.
In Lemma \ref{ob}, we provide an example for which $F_{\aaa}$ is $\Theta(n^2)$. 
Therefore, the number of sub-arcs in the lower envelope is lower bounded by $\Omega(n^2)$. 
This lower bound, along with the complexity of calculating the sub-arcs, inspired us to develop an approximation algorithm for this problem.
% , which is presented in Section \ref{appA}.

\begin{lemma}\label{ob}
$F_{\aaa}$ is lower bounded by $\Omega(n^2)$.
\end{lemma}

\begin{figure}[t]
\captionsetup[subfigure]{justification=centering}
\centering
  \begin{subfigure}{\linewidth}\centering
    \includegraphics[width=240pt]{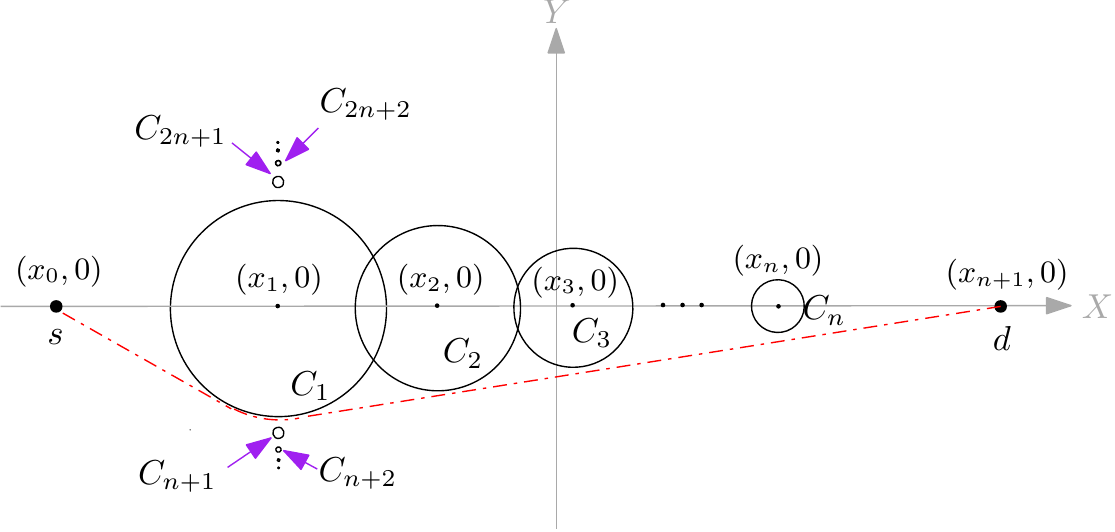}
    \caption{}
  \end{subfigure}  
\par\bigskip
  \begin{subfigure}{\linewidth}\centering
    \includegraphics[width=240pt]{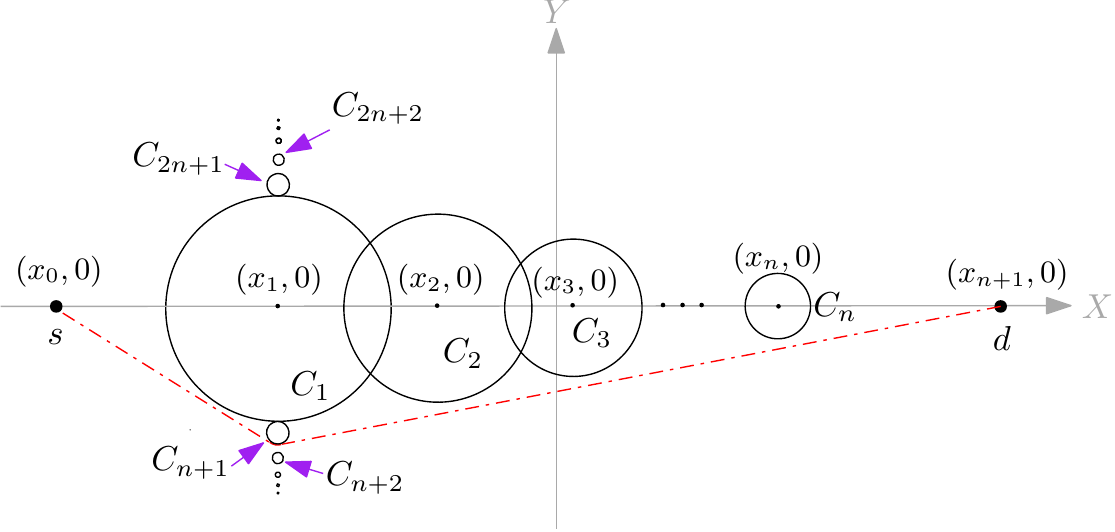}
    \caption{}
  \end{subfigure}  

    \caption{This figure illustrates an example where there are $\Theta(n^2)$ unique time-minimal paths for different departure times. For the ease of demonstration, we assumed $V \ll \vr$.} 
    \label{example2}
\end{figure}

\begin{proof}

We prove this lemma by giving an example where $F_{\aaa}$ is $\Theta(n^2)$.
 In this example, which is illustrated in Figure \ref{example2}, the source is located at $(x_0, 0)$ and the destination is located at $(x_{n+1}, 0)$, where $x_0 \ll x_{n+1}$. A set of growing discs $\mathscr{C}_1=\{C_1, ..., C_n\}$ are sorted along the x-axis, such that $C_i \in \mathscr{C}_1$ is centered at point $(x_i, 0)$, where $x_{i-1} < x_i < x_{i+1}$ and $x_{0} \ll x_i \ll x_{n+1}$. We also define a set of ``disjoint'' growing discs $\mathscr{C}_2=\{C_{n+1}, ..., C_{2n}\}$, such that $C_i \in \mathscr{C}_2$ is centered at point $(x_1, y_i)$, where $y_{i+1} < y_{i} < 0$. Similarly, we define a set of ``disjoint'' growing discs $\mathscr{C}_3=\{C_{2n+1}, ..., C_{3n}\}$, such that $C_i \in \mathscr{C}_3$ is centered at point $(x_1, -y_i)$. For the sake of simplicity, we assume that $V \ll \vr$, where $V$ is the speed of the discs and $\vr$ is the maximum speed of the robot.

Let $\mathcal{T}_1 = \{\tau_{1}, ..., \tau_n\}$ be a set of departure times, sorted in increasing order. Denote by $\{\T_{\tau_1}, ..., \T_{\tau_n}\}$ the set of time-minimal robot-paths corresponding to the departure times in $\mathcal{T}_1$.
We choose the initial radius of disc $C_1$ large enough (with respect to the other discs in $\mathscr{C}_1$) so that the robot-path $\T_{\tau_1}$ is tangent to disc $C_1$ and does not touch the other discs in $\mathscr{C}_1$ (see Figure \ref{example2} (a)). Similarly, we choose the initial radius of disc $C_2$ such that $\T_{\tau_2}$ is only tangent to $C_1$ and $C_2$. 
We repeat the same procedure for all the departure times in $\mathcal{T}_1$. Consequently, each robot-path $\T_{\tau_i}$ is tangent all the discs in $\{C_{1}, C_{2}, ..., C_i \}$. 

Observe that we can choose an appropriate value for $y_{n+1}$  such that $C_1$ and $C_{n+1}$ intersect at time $\tau'_{1}$, where for a small value of $\alpha$ we have $\tau_{n} < \tau'_{1} < \tau_{n} + \alpha$. Let $\T_{\tau'_1}$ be the time-minimal robot-path corresponding to the departure time $\tau'_1$.
As illustrated in Figure \ref{example2} (b), observe that $\T_{\tau'_1}$ is tangent to disc $C_{n+1}$. By an argument similar to the above paragraph, we can define a set $\mathcal{T}_2 = \{\tau'_{1}, ..., \tau'_n\}$, where the corresponding time-minimal paths of the departure times $\tau'_{i} \in \mathcal{T}_2$ are tangent to all the discs in $\{C_{n+1}, C_{2}, ..., C_i \}$. We repeat the above procedure for all the discs in $\mathscr{C}_2$. Let $\mathcal{T} = \mathcal{T}_1 \cup \mathcal{T}_2 \cup ... \cup \mathcal{T}_n$.

Let $\pi(s, d, \tau_i)$ and $\pi(s, d, \tau_j)$ be a pair of shortest paths in the adjacency graph $G$, corresponding to the departure times $\tau_i,\tau_j \in \mathcal{T}_1$, where $i \not= j$. Since $\T_{\tau_i}$ and $\T_{\tau_j}$ are tangent to different sets of discs, we have $\pi(s, d, \tau_i) \not= \pi(s, d, \tau_j)$. Thus, observe that $\aaa(t)$ consists of $\Theta(|\mathcal{T}|) =\Theta(n^2)$ sub-arcs.
\end{proof}

\section{Approximating $\aaa(t)$} \label{appA}
\subsection{The reverse shortest path}\label{reverse}
Let us define a function $\aaa^{-1}:[0, T_{max}] \rightarrow [0, T_{max}]$ where $\aaa^{-1}(t)$ is the latest departure time at $s$, when the robot arrives at $d$ at time $t$. In this section, we describe our method for computing the function $\aaa^{-1}(t)$ for a set of ``fixed''  values of $t$. We generalize the time-minimal path algorithm presented in \cite{shortest_path_in_growing_disc_yi} to the case where the discs are shrinking. 

A growing disc $C_i$ is defined by a pair $(O_i, R_i(t))$, where $O_i$ is the center and $R_i(t)$ is the radius of $C_i$ at time $t$. Let $\partial C_i(t)$ denote the boundary of the disc $C_i$ at time $t \in [0, T_{max}]$.
We now define a \textit{shrinking disc} $\hat{C_i}$ by a pair $(O_i, \hat{R_i}(t))$, such that $\hat{R_i}(t) = R_i(T_{max} - t)$. Note the following properies: 
\begin{itemize}
\item $\hat{C_i}$ and $C_i$ are centered at the same point.
\item $\partial C_i(0) = \partial \hat{C_i}(T_{max})$.
\item $\partial C_i(t) = \partial \hat{C_i}(T_{max} - t)$.
\end{itemize}

Let $\hat{\mathscr{C}}=\{\hat{C_1}, ..., \hat{C_n}\}$ be a set of shrinking discs. We begin with the following observation.
% , which is derived from the above properties.

\begin{observation}\label{obs1}
For any two given times $\tau$ and $\hat{\tau}$, where $0 \leq \tau < \hat{\tau} \leq T_{max}$, $ C_i(\tau)$ and $ C_j(\hat{\tau})$ have the same tangent lines as $ \hat{C_i}(T_{max} - \tau)$ and $ \hat{C_j}(T_{max} - \hat{\tau})$.
\end{observation}

% Corresponding to the discs in $\hat{\mathscr{C}}$, let: (1) $\ee$ be the set of all spiral and tangent paths; (2) $\steinerpoints$ be a set of all Steiner points. Recall that in Section \ref{prel}, we defined the adjacency graph $G = (V_s, E_s)$ of the growing discs, based on sets $\ee$ and $\steinerpoints$. Similarly, we define the adjacency graph $\hat{G} = (V_s, E_s)$ for the shrinking discs, with the same sets of vertices and edges. \hl{Thus, the only difference between $\hat{G}$ and $G$ is the weights of the edges.}

% The length of the tangent lines  and the location of the Steiner points are calculated differently for shrinking discs. 

% Note that, although $\hat{G}$ has the same vertices and edges as $G$, the path function (defined in preliminaries) is defined differently.

% By the definition, for each Steiner point $v \in \steinerpoints$ there is a unique vertex in $V_s$ denoted by $\dot{v}$.
% For each path $\overrightarrow{uv} \in \ee$ there is a unique edge $(\dot{u}, \dot{v})$ in $E_s$. The weight of the edges in $G$ are determined using the path function $\robotpath$. If the robot-path $\robotpath(e, \tau)$ is valid then the weight of $(\dot{u}, \dot{v})$ is equal to the length of $\robotpath(e, \tau)$, denoted by $|\robotpath(e, \tau)|$. If the robot path is invalid then the weight is $\infty$.
% Similarly, 

Let $\ell^{rl}_{ij}(\tau) = \overline{pq}$ be a tangent line where $p \in \partial C_i(\tau)$ and $q \in \partial C_j(\hat{\tau})$. Similarly, let $\hat{\ell}^{lr}_{ji}(T_{max} - \hat{\tau}) = \overline{qp}$ be a tangent line where $q\in \partial \hat{C_j}(T_{max} - \hat{\tau})$ and $p \in \partial \hat{C_i}(T_{max} - \tau)$.
 By Observation \ref{obs1}, $\ell^{rl}_{ij}(\tau)$ is equivalent to $\hat{\ell}^{lr}_{ji}(T_{max} - \hat{\tau})$. 
Thus, the two tangent paths $\overrightarrow{\hat{\ell}^{lr}_{ji}}(T_{max} - \hat{\tau})$ and $\overrightarrow{\ell^{rl}_{ij}}(\tau)$ are the same line segments, but with opposite directions.
 % Thus, for a tangent path from $C_i$ to $C_j$, along the tangent line $\ell^{rl}_{ij}(\tau)$, there exists a tangent path from $\hat{C_j}$ to $\hat{C_i}$ (along the tangent $\hat{\ell}^{lr}_{ji}(T_{max} - \hat{\tau})$) with equal length and in opposite direction. 
 We call $\overrightarrow{\hat{\ell}^{lr}_{ji}}(T_{max} - \hat{\tau})$ the \textit{reverse tangent path} of $\overrightarrow{\ell^{rl}_{ij}}(\tau)$.
 Similarly, for a spiral path $\overrightarrow{\sigma}(\tau)$ there exists a \textit{reverse spiral path} $\overrightarrow{\hat{\sigma}}(T_{max} - \hat{\tau})$.
Moreover, we can extend this definition to a robot-path: let $\T$ be a valid robot-path from $s$ to $d$, where the departure time is $\tau_s$ and the arrival time is $\tau_d$. Then, there exists a \textit{reverse robot-path} $\hat{\T}$ from $d$ to $s$ whose departure time is $T_{max} - \tau_d$ and arrival time is $T_{max} - \tau_s$.

% /
Recall that in Step ($ii$) of the SPGD algorithm (see Section \ref{PTDSP}), the adjacency graph $G$ is constructed using the identified tangents and spiral paths in Step ($i$). 
Similarly, we construct the \textit{reverse adjacency graph} $\hat{G}$ using the reverse tangents and spiral paths.

\begin{lemma}\label{reverse-path}
Let $\pi(u, v, \tau_u)$ be a valid path from vertex $u$ to vertex $v$ in $G$, where the departure time is $\tau_u$ and the arrival time is $\tau_v$. Then, there exists a valid path $\hat{\pi}(v, u, T_{max} - \tau_v)$ in $\hat{G}$ whose arrival time is $T_{max} - \tau_u$.
\end{lemma}
\begin{proof} 
Since $\pi(u, v, \tau_u)$ is a valid path, there exists a robot-path $\T$ with departure time $\tau_u$ and arrival time $\tau_v$. By definition, there exists a reverse robot-path of $\T$, denoted by $\hat{\T}$, whose departure time and arrival time are $T_{max} - \tau_v$ and $T_{max} - \tau_u$, respectively. Thus, there exists a valid path $\hat{\pi}(v, u, T_{max} - \tau_v)$ in $\hat{G}$. 
 \end{proof}

% Let $\pi(s, d, \tau)$ be a path in $G$ whose arrival time is $\hat{\tau}$. By the above lemma, there exists a \textit{reverse path} $\hat{\pi}(v_h, v_1, T_{max} - \hat{\tau})$ in $\hat{G}$ whose arrival time is $T_{max} - \tau$. 

% \begin{theorem}\label{reverse-theorem}
% For a departure time $\tau$, let $\hat{\tau}$ be the arrival time of a shortest path from $s$ to $d$ in $G$. Then, $\tau$ is the arrival time of a reverse shortest path of $(s, d)$, when its departure time is $\hat{\tau}$. 
% \end{theorem}
% \begin{proof} This is a direct consequence of Lemma \ref{reverse-path}. \end{proof}

By the above lemma, for any path $\pi$ from $s$ to $d$ in $G$, there exists a path $\hat{\pi}$ from $d$ to $s$ in $\hat{G}$ of the same length as $\pi$. Thus, in order to find a shortest path form $s$ to $d$, we can find a shortest path in  $\hat{G}$ and reverse its direction. Therefore, similar to the SPGD algorithm in Section \ref{PTDSP}, a reverse shortest path can be found by running Dijkstra's algorithm in $\hat{G}$. We summarize the steps of the above algorithm as follows.

% \hlr{Should I make a lemma out of the above paragraph?}

\begin{itemize}
\item[($i$)] Identify the reverse tangents and spiral paths.
\item[($ii$)] Construct the reverse adjacency graph $\hat{G}$.
\item[($iii$)] Run Dijkstra's algorithm on $\hat{G}$ to find a time-minimal path.
\end{itemize}

We call the above algorithm the \textit{reverse shortest path among growing discs} (RSPGD). Using this algorithm, 
 for any given time $t$, a time-minimal path can be found which arrives at destination at time $t$. 
 
\begin{corollary}\label{con-a-}
For a given arrival time $t$ at $d$, $\aaa^{-1}(t)$ can be computed by running the RSPGD algorithm.
\end{corollary}

We should remark that finding a shortest path from $d$ to $s$ in $\hat{G}$ does not always yield a time-minimal path among shrinking discs. For example, consider the case where the robot stops and waits for some discs to shrink to a certain size, until they open a previously blocked path. Then, the robot starts moving with maximum velocity towards the destination along the recently opened path. This contradicts our assumption that the robot always moves with the maximum speed. Thus, a shortest path in $\hat{G}$ does not guarantee a time-minimal robot-path among shrinking discs. 

% Let $\hat{\robotpath}:E_s \times T \rightarrow {\rm I\!R}^2$ be a function where $\hat{\robotpath}(\overrightarrow{\dot{v}\dot{u}}, \tau)$ is the robot-path $\overrightarrow{vu}$, when the departure time is $\tau$. 

\subsection{Approximation Algorithm} \label{algo}

In this section, we present a $(1+\epsilon)$-approximation algorithm for computing the minimum arrival time function $\aaa(t)$. 
To obtain an approximation for $\aaa(t)$, our algorithm (Algorithm \ref{compute-b}) computes a set of arrival time values
% we first define a set of departure times whose corresponding arrival times are spaced $\epsilon$ apart from each other.
% Algorithm \ref{compute-b} determines a set of valid arrival time values, denoted by 
$A=\{a_1, a_2, ..., a_m\}$, such that, for $2 \leq i \leq m$, ${a_i \over a_{i-1}} = 1+\epsilon$ (i.e., arrival times are spaced within a factor of $1+\epsilon$ from each other).  Since $\aaa(t)$ is an increasing function (refer to Lemma \ref{T-properties-1}), for each valid arrival time, there exists a unique departure time. This is a key distinction to some of the previously standard variants of the time-dependent shortest path problems.
For each $a_i \in A$, the algorithm runs the RSPGD algorithm to find its corresponding departure time $b_i$. Let us denote the departure time values by a set $B=\{b_1, b_2, ..., b_m\}$. Each departure time in $B$ is referred to as a \textit{sampled time}.

\begin{algorithm}[H]
\caption{Computing $B$} \label{compute-b}
\begin{algorithmic}[1]

\STATE $B = \emptyset$, $A = \emptyset$, $i = 0$, $a_0 = \aaa(0)$ \\ \COMMENT{ $\aaa(0)$ is calculated by running the SPGD algorithm \cite{shortest_path_in_growing_disc_yi}}
\WHILE {$a_i < T_{max}$}
\STATE $a_{i + 1} = (1+\epsilon)a_{i}$ 
\STATE $b_i = \aaa^{-1}(a_{i + 1})$
\STATE $B := B \cup \{b_i\}$
\STATE $A := A \cup \{a_{i + 1}\}$
\STATE $i = i + 1$
\ENDWHILE
\STATE \textbf{return} $B$ and $A$

\end{algorithmic}
\end{algorithm}

\begin{lemma}\label{tmp}
Algorithm \ref{compute-b} runs $O({1 \over \epsilon}\log({\vr \over \vm}))$ time-minimal path computations.
\end{lemma}
\begin{proof} 
We first estimate the number of sampled times in $B$. Let $B=\{b_1, b_2, ..., b_m\}$. By definition we have 
\begin{align*}   
{\aaa({b_m}) \over \aaa({b_1})} &= {a_m \over a_1} =(1+\epsilon)^{m-1} 
\end{align*} 
By Lemma \ref{T-properties-2}, for any $1 \leq i \leq k$ we have $\aaa(b_i) \leq {|\overline{sd}| \over \vr}$ and $\aaa(b_i) \geq {|\overline{sd}| \over \vm}$. So, we obtain
\begin{align*}(1+\epsilon)^{m-1} \leq {\vr \over \vm}\end{align*}
For $\epsilon \in (0, 1)$, we observe that ${\epsilon \over 2} < \log (1+\epsilon)$. Thus, 
\begin{align*}(m-1) {\epsilon \over 2} \leq \log\Big({\vr \over \vm}\Big)\end{align*}
 Therefore, there are $O({1 \over \epsilon}\log({\vr \over \vm}))$ sampled times in  $B$. For each sampled time, the algorithm runs an instance of the reverse shortest path algorithm in Line 5. Thus, the total number of time-minimal path computations is $O({1 \over \epsilon}\log({\vr \over \vm}))$.
\end{proof}
Using the sampled times reported by Algorithm \ref{compute-b}, we now define a step function $\aae : [0, T_{max}] \times (0,1) \rightarrow A$ such that for $t \in [b_i, b_{i+1})$, we have $\aae(t, \epsilon) = a_{i+1}$.
\begin{lemma}\label{approximation}
For any real constant value of $\epsilon \in (0, 1)$, function $\aae$ is a $(1+\epsilon)$-approximation for the arrival time function $\aaa$.
\end{lemma}
\begin{proof} 
 By definition, for any $t \in [b_{i}, b_{i+1})$ we have $\aae(t, \epsilon) = a_{i+1}$. Referring to the fact that $\aaa$ is an increasing function, for any $t \in [b_{i}, b_{i+1})$ we obtain $\aaa(b_{i}) \leq \aaa(t) < \aaa(b_{i+1})$. Thus,
\begin{align*}
{\aaa(b_{i+1}) \over \aaa(b_{i+1})} < { \aae(t, \epsilon) \over \aaa(t)} \leq {\aaa(b_{i+1}) \over \aaa(b_{i})}  
\end{align*}
Since we have $\aaa(b_{i+1}) = a_{i+1}$ and $\aaa(b_{i}) = a_{i}$:
\begin{align*}  
1 < {\aae(t, \epsilon) \over \aaa(t)} \leq  {a_{i+1} \over a_{i}} = 1 + \epsilon 
\end{align*}
\end{proof}

\begin{theorem}\label{tmp-all}
The minimum arrival time function can be approximated by executing $O({1 \over \epsilon}\log({\vr \over \vm}))$ time-minimal path computations.
\end{theorem}

\begin{proof}
This is a direct result of Lemmas \ref{tmp} and \ref{approximation}. 
\end{proof}

Since the time-minimal path algorithm for fixed departure times runs in $O(n^2 \log (n))$ time \cite{shortest_path_in_growing_disc_yi}, the time complexity of our preprocessing algorithm is $O({n^2 \over \epsilon}\log({\vr \over \vm})\log(n))$. For a given query departure time $t \geq 0$, we can report the approximated value of the minimum arrival time (i.e., $\aae(t, \epsilon)$) using a binary search in $B$. 
 In Lemma \ref{tmp}, we proved that the size of the set $B$ is $O({1 \over \epsilon}\log({\vr \over \vm}))$. Thus, the query time of our algorithm is $O(\log ({1 \over \epsilon}) + \log\log({\vr \over \vm}))$.

\section{Conclusions}\label{conclusions}
In this paper, we studied the time-dependent minimum arrival time problem among growing discs. We
presented a $(1+\epsilon)$-approximation to compute the minimum arrival time function. Our algorithm runs  shortest path algorithms as a black-box and  its time complexity  is determined by the number of such calls. 
% In this paper, we studied the same-speed growing discs setting in which a shortest path can be found in $O(n^2 \log n)$ time \cite{shortest_path_in_growing_disc_yi}.
Therefore, for different problem settings, we can plug-in different shortest path algorithms.
For example, Nouri and Sack \cite{nouri} studied a variant of the SPGD problem where the growth rates of the discs are given as polynomial functions of degree $\dd$. In this algorithm a query time-minimal path can be found in $O(n^2 \log (\dd n))$ time. By plugging-in this algorithm, our preprocessing step executes $O({1 \over \epsilon}\log({\vr \over \vm}))$  shortest path computations, running in $O({n^2 \over \epsilon}\log({\vr \over \vm})\log(\dd n))$ time.
% \hlr{Should I create a new section for the above generalized version? In that case we can also update the lower bound for $F_{\aaa}$.}

In order to compute the output size of the minimum arrival time function $\aaa(t)$, one would need to determine the number of sub-arcs in the lower envelope, denoted by $F_{\aaa}$. We presented a lower bound on $F_{\aaa}$ and we leave as an open problem to establish an upper bound. 

Another interesting open problem is to approximate the minimum arrival time function when the query involves the departure time at $s$, as well as
$s$ and $d$ as part of the input.
% Here, we studied the case where query 
% This problem has been previously studied for the instance where the departure times are fixed}

% The description complexity of the function $\aaa(t)$ is affected by the description complexities of sub-arcs in the lower envelope. We presented a lower bound on the number of intersections between each pair of arrival curves and we leave as an open problem to establish an upper bound. 

\bibliographystyle{unsrt}
\bibliographystyle{plainurl}
\bibliography{thesis}

\end{document}